\documentclass[final]{dmtcs-episciences}

\usepackage[utf8]{inputenc}
\usepackage{subfigure}
\usepackage{amssymb,amsmath,amsthm}
\usepackage[T1]{fontenc}
\usepackage[round,sort&compress]{natbib}

\newtheorem{theorem}{Theorem}
\newtheorem{lemma}[theorem]{Lemma}
\newtheorem{corollary}[theorem]{Corollary}

\author{Marcin Pilipczuk\affiliationmark{1}\thanks{Supported by Polish National Science Centre SONATA BIS-12 grant number 2022/46/E/ST6/00143.}
  \and Paweł Rzążewski\affiliationmark{1,2}\thanks{This work is a part of the project BOBR that has received funding from the European Research Council (ERC) under the European Union's Horizon 2020 research and innovation programme (grant agreement No.~948057)}  }
\title[Minimal separators and potential maximal cliques in $P_6$-free graphs of bounded clique number]{A polynomial bound on the number of minimal separators and potential maximal cliques in $P_6$-free graphs of bounded clique number}

\affiliation{  
  Institute of Informatics, University of Warsaw, Poland\\
  Warsaw University of Technology, Poland}
\begin{document}
\publicationdata{vol. 27:2}{2025}{10}{10.46298/dmtcs.12438}{2023-10-19; 2023-10-19; 2025-02-17}{2025-03-12}
\maketitle
\begin{abstract}
  In this note we show a polynomial bound on the number of minimal separators and potential maximal cliques in $P_6$-free graphs of bounded clique number.
\end{abstract}

\section{Introduction}

Let $G$ be a graph.
For a set $X \subseteq V(G)$, we say that a connected component 
$C$ of $G-X$ is a \emph{full component} of $X$ if $N_G(C) = X$, that is,
every vertex of $X$ has a neighbor in $C$. 
A set $X \subseteq V(G)$ is a \emph{minimal separator} 
if it has at least two full components. 

Intuitively, the space of all minimal separators of a graph $G$
reflects the space of all possible separations that can be used
to solve some computational problem on $G$ via dynamic programming. 
A related notion of \emph{potential maximal clique}
(not defined formally in this work) corresponds to all reasonable
choices of a bag in a tree decomposition of $G$. 

\cite{BouchitteT01}
(with some results generalized by \cite{FominTV15}) showed that indeed these notions 
are useful to solve a wide family of graph problems.

\begin{theorem}[\cite{BouchitteT01,BouchitteT02}]\label{thm:minsep-pmc}
If $G$ is an $n$-vertex graph with $a$ minimal separators and
$b$ potential maximal cliques, then $b \leq n(a^2 + a + 1)$ and $a \leq nb$.
Furthermore, given a graph $G$, one can in time polynomial in 
the input and compute the list of all its minimal separators and potential maximal
cliques.
\end{theorem}

\begin{theorem}[\cite{BouchitteT01,FominTV15}, informal statement]\label{thm:algo}
    A wide family of graph problems, including \textsc{Maximum Weight Independent Set} and \textsc{Feedback Vertex Set}, can be solved in time polynomial in the size of the input graph and the number of its potential maximal cliques.
\end{theorem}

As a result, if for a graph class $\mathcal{G}$ one can prove 
a (purely graph-theoretical) polynomial bound on the number of
minimal separators or potential maximal cliques in $\mathcal{G}$,
then one can immediately obtain polynomial-time algorithms
for a wide family of problems on $\mathcal{G}$. 
We call such a graph class \emph{tame}. 
Recently, a methodological study of which graph classes are tame was initiated, see, e.g., the work of~\cite{GartlandL23a,GartlandL23b,AbrishamiCDTTV22,MilanicP21,GajarskyJL0PRS22}.

We say that a graph $G$ is \emph{$H$-free} for a graph $H$ if no induced
subgraph of $G$ is isomorphic to $H$. 
The \emph{clique number} of a graph $G$ is the maximum cardinality
of a set pairwise adjacent vertices of $G$. 
For an integer $t$, by $P_t$ we denote the path on $t$ vertices (and
$t-1$ edges). 
The main result of this note is a short proof that $P_6$-free graphs
of bounded clique number have a polynomial number of minimal separators
and potential maximal cliques.

\begin{theorem}\label{thm:main}
Let $G$ be an $n$-vertex $P_6$-free graph of clique number $k$.
Then, $G$ contains at most $(2n)^{k+1}$ minimal separators and
at most $2^{2k+2} n^{2k+3}$ potential maximal cliques.
\end{theorem}
We remark that some additional condition to just $P_6$-freeness
is needed for the conclusion of Theorem~\ref{thm:main}, as 
even the class of $P_5$-free graphs is not tame. This can be witnessed
by \emph{prisms}: An $n$-prism consists of two $n$-vertex cliques
with a matching in between; it is $P_5$-free but admits
$2^n-2$ minimal separators.
Furthermore, the $P_6$ cannot be replaced with $P_7$:
\cite{CMPPR23} provide a construction
of $P_7$-free graphs $G_n$ that have clique number $2$,
$|V(G_n)| = 6n+2$, and at least $3^n$ minimal separators. 

Our motivation for proving Theorem~\ref{thm:main} is two-fold.
First, the proof is very simple, much simpler than the 
arguments of~\cite{GrzesikKPP22} giving polynomial-time algorithms for \textsc{Maximum
Weight Independent Set} and related problems in $P_6$-free graphs
without any assumption on the clique number. 
Second, it gives examples of graph classes that are tame,
but in which such problems as \textsc{$5$-Coloring}
or \textsc{Odd Cycle Transversal} are NP-hard:
\cite{Huang16} proved that \textsc{$5$-Coloring} is \textsf{NP}-hard in $P_6$-free graphs
(and trivial in graphs of clique number larger than $5$)
and \cite{DabrowskiFJPPR20,ChudnovskyKPRS21} proved that \textsc{Odd Cycle Transversal} is \textsf{NP}-hard
in $P_6$-free graphs of clique number at most $3$.
This answers negatively a question of~\cite[Open problem~3]{MilanicP21}.

\section{Proof of Theorem~\ref{thm:main}}

If $G$ is edgeless, then the statement is immediate (there are
no minimal separators and $n$ potential maximal cliques), so 
we assume $E(G) \neq \emptyset$ and thus $n,k \geq 2$. 

Assuming $n,k \geq 2$, we will prove that $G$ contains at most $(2n)^k \cdot (2n-1)$ minimal
separators. This directly implies the bound for minimal separators of Theorem~\ref{thm:main}
and also implies the bound on the number of potential maximal cliques in $G$ via 
Theorem~\ref{thm:minsep-pmc}.

\medskip

We need \emph{modules} and some basic facts about them.
Let $G$ be a graph. 
A \emph{module} is a nonempty set $M \subseteq V(G)$
such that every $u \in V(G) \setminus M$ is adjacent to either
all vertices of $M$ or to no vertex of $M$. 
A module $M$ is:
\begin{itemize}
    \item \emph{strong} if for every other module $M'$
either $M \subseteq M'$, $M' \subseteq M$, or $M \cap M' = \emptyset$;
    \item \emph{proper} if it is different than $V(G)$;
    \item \emph{maximal} if it is proper and strong and no other proper
    strong module contains it;
    \item connected if $G[M]$ is connected.
\end{itemize}

We need the following properties of modules,
which can be distilled from Lemma~2 and Theorem~2 and the discussion
around them in the work of~\cite{HabibP10}.
\begin{lemma}\label{lem:module-partition}
In a graph on at least two vertices, the maximal modules form a partition of the vertex set.
\end{lemma}
\begin{lemma}\label{lem:modules}
An $n$-vertex graph $G$ contains at most $2n-1$ strong modules.
Furthermore, for every module $M$ in $G$, 
there exists a unique strong module $M'$ that contains $M$ and is inclusion-wise minimal with this property, and one of the following holds:
\begin{itemize}
\item $M = M'$ and $M$ is a strong module in $G$;
\item $G[M']$ is disconnected and $M$ is a union of at least two connected components of $G[M']$;
\item the complement of $G[M']$ is disconnected and $M$ is a union of at least two connected components of the complement of $G[M']$.
\end{itemize}
\end{lemma}
In the context of graphs of bounded clique number, we observe
the following immediate corollary.
\begin{corollary}\label{cor:modules}
  An $n$-vertex graph of clique number $k$ contains at most
  $2^k (2n-1)$ connected modules. 
\end{corollary}
\begin{proof}
    Fix a connected module $M$ and let $M'$ be as in Lemma~\ref{lem:modules}.
    Then, Lemma~\ref{lem:modules} implies that $M$ is a union of some connected components
    of the complement of $G[M']$ (possibly all of them for the first option of Lemma~\ref{lem:modules}).
    There are at most $(2n-1)$ choices for $M'$ and at most $2^k$ choices which components
    of the complement of $G[M']$ form $M$, as the clique number of $G$ is $k$.
\end{proof}

\medskip

We will also need the following lemma of~\cite{GrzesikKPP22}.
\begin{lemma}[cf. Lemma~4.2 of~\cite{GrzesikKPP22}]\label{lem:nei}
Let $G$ be a graph, $X \subseteq V(G)$, and let $A$ be a full
component of $X$ with $|A| > 1$.
Let $p,q \in A$ be any two vertices in different maximal
modules of $G[A]$ (which exist by Lemma~\ref{lem:module-partition}) that are adjacent
(i.e., $pq \in E(G)$). 
Then, for every $x \in X$ either:
\begin{itemize}
    \item There exists an induced $P_4$ in $G$ with one endpoint
    in $x$ and the remaining three vertices in $A$.
    \item The vertex $x$ is adjacent to $p$ or to $q$ (or to both).
    \item The complement of $G[A]$ is disconnected
    and $N(x) \cap A$ consists of some connected components
    of the complement of $G[A]$. 
\end{itemize}
\end{lemma}
We remark that if the complement of $G[A]$ is disconnected,
then the connected components of the complement of $G[A]$ are 
exactly the maximal proper strong modules of $G[A]$. 

We deduce the following.
\begin{lemma}\label{lem:dom}
Let $G$ be a $P_6$-free graph of clique number $k \geq 2$,
let $X$ be a minimal separator in $G$, and let $A$ and $B$
be two distinct full sides of $X$. 
Then there exists a set $Q \subseteq A$ of size at most $k$
such that every vertex of $X \setminus N(Q)$ is complete to $B$.
\end{lemma}
\begin{proof}
If $|A| = 1$, take $Q = A$, so $X \setminus N(Q) = \emptyset$ and we are done.
Assume then $|A| > 1$.

If the complement of $G[A]$ is disconnected, take $Q$ to be any
set consisting of one vertex from each connected component
of the complement of $G[A]$. 
Otherwise, take $Q$ of size $2$, consisting of any two vertices
of two different maximal modules of $G[A]$ that are adjacent
(such two modules exist as $|A| > 1$ and $G[A]$ is connected).
Since the clique number of $G$ is $k$, we have $|Q| \leq \max(k,2) = k$.
By Lemma~\ref{lem:nei}, for every $x \in X \setminus N(Q)$ there exists
a $P_4$ with one endpoint in $x$ and the remaining three vertices in $A$;
denote it $P^x$.
If $x$ is not complete to $B$, then, as $G[B]$ is connected and
$B$ is a full component of $X$, there are $y,z \in B$ with $xy,yz \in E(G)$
but $xz \notin E(G)$. Then, $P^x$, prolonged with $y$ and $z$
is an induced $P_6$ in $G$, a contradiction.
\end{proof}

By Lemma~\ref{lem:dom}, we infer that $B$ is a connected module
of $G-N(Q)$. As $|Q| \leq k$, there are at most $n^k$ choices of $Q$.
By Corollary~\ref{cor:modules}, for fixed $Q$, there are at most
$2^k (2n-1)$ choices of $B$. As $X = N_G(B)$, there are at most
$(2n)^k \cdot (2n-1)$ minimal separators in $G$,
as promised.

\acknowledgements
We are grateful to Jakub Gajarsk\'y, Lars Jaffke, Paloma Thom\'e de Lima, Jana Masa\v{r}ikov\'a, U\'everton Souza, and Michał Pilipczuk for inspiring discussions and to Martin Milani\v{c} for bringing our attention to Open problem 3 of ~\cite{MilanicP21}.

\nocite{*}
\bibliographystyle{abbrvnat}
\bibliography{references}
\label{sec:biblio}

\end{document}